\documentclass[11pt]{article}
\usepackage[affil-it]{authblk}
\usepackage{amssymb,latexsym,amsmath}
\usepackage{amsthm}

\theoremstyle{definition}

\newtheorem{theorem}{Theorem}
\newtheorem{corollary}[theorem]{Corollary}
\newtheorem{proposition}[theorem]{Proposition}
\newtheorem{lemma}[theorem]{Lemma}
\newtheorem{definition}[theorem]{Definition}
\newtheorem{example}[theorem]{Example}
\newtheorem{notation}[theorem]{Notation}
\newtheorem{remark}[theorem]{Remark}

\usepackage{graphicx}

\usepackage{relsize}

\newcommand{\numberset}{\mathbb}

\newcommand{\N}{\numberset{N}}
\newcommand{\Z}{\numberset{Z}}

\newcommand{\K}{\numberset{K}}
\newcommand{\F}{\numberset{F}}

\newcommand{\mS}{\mathcal{S}}
\newcommand{\mC}{\mathcal{C}}
\newcommand{\mD}{\mathcal{D}}

\newcommand{\mG}{\mathcal{G}}

\usepackage{hyperref}
\newcommand{\mA}{\mathcal{A}}

\newcommand{\mU}{\mathcal{U}}

\usepackage[margin=2.6cm]{geometry}

\newcommand{\mF}{\mathcal{F}}

\title{\textbf{Generalized weights:  an anticode approach}}

\author{Alberto Ravagnani%
  \thanks{E-mail: \texttt{alberto.ravagnani@unine.ch}. \ The author was
partially supported by the Swiss National Science
Foundation through grant no. 200021\_150207.}}
\affil{\textnormal{Institut de Math\'{e}matiques} \\ \textnormal{Universit\'{e}
de Neuch\^{a}tel}\\\textnormal{Emile-Argand 11, CH-2000 Neuch\^{a}tel,
Switzerland}}

\date{}

\begin{document}

\maketitle

\begin{abstract}
In this paper we study generalized weights as an algebraic
invariant of a code. We first describe anticodes in the Hamming and in the rank
metric,
proving in particular that optimal anticodes in the rank metric coincide with 
Frobenius-closed spaces. Then we characterize both generalized Hamming and rank
weights of a code in terms of the intersection of the
code with optimal anticodes in the respective metrics. Inspired by this
description, we propose a new algebraic invariant, which we call
``Delsarte generalized weights'', for
Delsarte rank-metric codes  based on optimal anticodes of matrices.
We show that our invariant refines
the generalized rank weights for Gabidulin codes proposed by
Kurihara, Matsumoto and Uyematsu, and establish a series 
of properties of Delsarte generalized weights. 
In particular, we characterize Delsarte optimal codes and anticodes
in terms of their generalized weights.
We also present a duality theory
for the new  algebraic invariant, proving that 
the Delsarte generalized weights of a code completely determine the Delsarte
generalized weights of the dual code.
Our results extend the theory of generalized rank weights for Gabidulin
codes. Finally, we
prove the analogue for Gabidulin codes of a theorem of Wei, proving that
their generalized rank weights characterize the
worst-case
security drops of a Gabidulin rank-metric code.
\end{abstract}

\maketitle

\section*{Introduction}\label{intr0}
Linear codes with the Hamming metric can be employed in wiretap channels to secure a communication against an
eavesdropper (see e.g. \cite{ozarow}). In \cite{wei}, Wei proved that 
in this context the
performance  of a code  
is measured by an algebraic invariant of the code, namely, the
collection
of its generalized Hamming weights.
Generalized Hamming weights have interesting mathematical properties.
For
example, they generalize the notion of minimum distance, 
and form a strictly increasing sequence
of integers. Another interesting combinatorial feature
is that the generalized Hamming weights
of a linear code completely determine the generalized Hamming weights of
the dual code. The generalized
Hamming weights of a code are defined in terms of the supports of the
subcodes of the code of given dimension.

Recently, Silva and Kschischang proposed a scheme based on Gabidulin
rank-metric codes
to secure a communication against an eavesdropper over a network in a universal
way (see \cite{metrics3} for details).
An important feature of the scheme is that it
is compatible with linear network coding.
Generalized  rank weights were introduced by Kurihara, Matsumoto
and Uyematsu in \cite{kmat} to
measure the performance of a Gabidulin code when employed in the scheme of \cite{metrics3}.
The generalized rank weights of a Gabidulin code are defined in terms of the
intersections of the code with Frobenius-closed spaces.
Generalized rank weights also have interesting mathematical properties,
including a duality theory (see \cite{kmat} and \cite{ducoat}).

In \cite{del1} Delsarte defines rank-metric codes as linear spaces of matrices
of given size over a finite field. There exists a natural way to
associate to a Gabidulin code
 a Delsarte code  with the same metric properties.
Thus Delsarte codes may be regarded as a generalization of Gabidulin
codes. 
It is not clear however how to extend the definition of generalized rank weights
for Gabidulin codes to Delsarte codes in a
convenient way, i.e., producing a well-behaving algebraic invariant. 
This is the main problem that we address in our work.

Gabidulin and Delsarte codes have interesting applications in coherent and
non-coherent
 network coding, e.g., they play an important role in the
construction of
subspace codes for random network coding in the approach of
\cite{KK1}. We address the interested reader to \cite{metrics}.

In this paper we focus on generalized weights for linear, Gabidulin and
Delsarte codes from an algebraic point of view. We first
investigate optimal anticodes
in the Hamming and in the rank metric, and show that both the generalized
Hamming weights and the 
generalized rank weights of a code can be characterized in terms of the
intersection of
the code with optimal anticodes in the respective metrics. In order
to establish this characterization for generalized rank weights, we prove in
particular that 
Frobenius-closed spaces in $\F_{q^m}^k$ coincide with optimal
anticodes
in the rank metric. The result says that the algebraic condition 
of being Frobenius-closed may be regarded as a metric condition. We
also give a convenient method to
compute a basis defined over $\F_q$ of a Frobenius-closed space $V \subseteq
\F_{q^m}^k$.

Inspired by the characterizations above, we propose a  definition of
generalized weights for Delsarte rank-metric codes
based on optimal anticodes in the space of matrices. Then we prove that Delsarte
generalized weights, as an algebraic invariant,
refine generalized rank weights for Gabidulin codes. We establish several properties of Delsarte generalized weights, which 
may be regarded as the analogue for Delsarte codes of the classical properties
of generalized Hamming and rank weights.
In particular, we show that Delsarte optimal codes and anticodes are
characterized by their Delsarte generalized weights.
We also study how Delsarte generalized weights relate to the duality theory of Delsarte codes. In particular, we
prove that the Delsarte generalized weights of a code  determine the Delsarte generalized weights of the dual code.

Finally, we show that the generalized rank weights 
proposed by Kurihara, Matsumoto and Uyematsu in \cite{kmat} measure the
worst-case security drops
of a Gabidulin code employed in the scheme of \cite{metrics3}.

The paper is organized as follows. In Section \ref{prel} we give preliminary definitions and results
on linear and rank-metric codes. In Section \ref{sss1} we characterize
generalized Hamming weights
in terms of optimal anticodes in the Hamming metric. In Section \ref{sss2} we prove that
Frobenius-closed spaces coincide with optimal anticodes in the rank metric, and characterize
generalized rank weights in terms of optimal anticodes. In Section \ref{sss4}
we 
introduce Delsarte codes and define Delsarte generalized weights,
proving that they refine generalized rank weights for Gabidulin codes. The main
properties
of Delsarte generalized weights are derived in Section \ref{sss5}. In
Section \ref{sss6}
we focus on the duality theory of Delsarte codes, showing that the
generalized weights of a Delsarte code
determine the generalized weights of the dual code. We prove the analogue for Gabidulin codes of a theorem of 
Wei on security drops in Section \ref{sss7}.

\section{Preliminaries}\label{prel}

In this section we briefly recall some basic notions of coding theory. In particular,
we give the definition of generalized weights for the Hamming and the rank
metric.

\begin{notation}
 Throughout this paper, $q$  denotes a prime power, and $\F_q$ the
finite field with $q$ elements. We also work with fixed positive integers
$n$, $k$ and $m$ with $k \le m$. For $s \in \N_{\ge 1}$, we set $[s]:=\{
1,2,...,s\}$, and if $\F$ is a field then the entries of a vector $v \in \F^s$
are denoted by 
$v_1,...,v_s \in
\F$. The vector space of matrices of size $t \times s$ over the field $\F$ is
$\mbox{Mat}(t\times s, \F)$, and if $M$ is any such matrix
we
denote by  $\mbox{rowsp}(M)$  the vector space generated over
$\F$ by the rows of $M$. If we work with a field extension $\K \supseteq \F$, to avoid confusion
we may also write $\mbox{rowsp}_{\K}(M)$ for the space generated over $\K$
by the rows of $M$. The rank of a matrix $M$ is $\mbox{rk}(M)$, while $M^t$
denotes the transpose of $M$. The trace of a square matrix
$M$ is $\mbox{Tr}(M)$.
\end{notation}

Let us start with classical codes in the Hamming metric.

\begin{definition} \label{def1}
  A \textbf{linear code} of length $n$ and dimension $t$ is a
$t$-dimensional $\F_q$-subspace $C \subseteq \F_q^n$. The
\textbf{Hamming} \textbf{weight} of a
vector $v \in \F_q^n$ is defined as $\mbox{wt}(v):= |\{ i \in [n] : v_i \neq
0\}|$. The \textbf{minimum weight} of a non-zero code $C$ is $\mbox{minwt}(C):=
\min \{ \mbox{wt}(c) : c \in C, \ c \neq 0 \}$, and the \textbf{maximum weight}
of any code $C$ is $\mbox{maxwt}(C) := \max \{ \mbox{wt}(c) : c \in C\}$.
The \textbf{support} of an $\F_q$-subspace $D \subseteq \F_q^n$ is defined by
$\chi(D):= \{ i \in [n] : \exists \ d \in D \mbox{ with } d_i \neq 0\}$. Given
a $t$-dimensional   non-zero code $C \subseteq \F_q^n$ and an integer
$1 \le r \le t$, the \textbf{$r$-th generalized Hamming weight}
of $C$ is
$d_r(C):= \min \{ |\chi(D)| : D \subseteq C, \ \dim_{\F_q}(D)=r\}$.
\end{definition}

 In \cite{wei} Wei proved that  generalized Hamming weights  characterize
the worst-case security drops of a linear code employed in the coding
scheme for wiretap channels
proposed in \cite{ozarow}. 
The main algebraic properties of generalized Hamming weights are summarized in
the
following result.

\begin{theorem}[see \cite{wei}] \label{propr1}
 Let $C \subseteq \F_q^n$ be a non-zero linear code of dimension $1 \le t
\le n$ over $\F_q$. The following hold.
 \begin{enumerate}
  \item $d_1(C)= \mbox{minwt}(C)$.
  \item $d_t(C) \le n$.
  \item For any $1 \le r \le t-1$ we have $d_r(C)< d_{r+1}(C)$.
  \item For any $1 \le r \le t$ we have $d_r(C) \le n-t+r$.
 \end{enumerate}
 \end{theorem}

We now introduce Gabidulin rank-metric codes and generalized rank weights.
Given a vector $v=(v_1,...,v_k) \in \F_{q^m}^k$, let
$v^q:=(v_1^q,...,v_k^q)$.
A subspace $V \subseteq \F_{q^m}^k$ is  \textbf{Frobenius-closed}
if $v \in V$ implies $v^q \in V$. We denote by $\Lambda_q(k,m)$ the
set of
Frobenius-closed spaces $V \subseteq \F_{q^m}^k$.

\begin{definition} \label{rgw}
 A \textbf{Gabidulin code} $C$ of length $k$
and dimension $t$ is an $\F_{q^m}$-subspace $C \subseteq \F_{q^m}^k$
of dimension $t$ over $\F_{q^m}$.
The
\textbf{rank} of any
vector $v \in \F_{q^m}^k$ is  $\mbox{rk}(v):=
\dim_{\F_q}
\mbox{Span}_{\F_q} \{ v_1,...,v_k\}$. The \textbf{minimum rank} of a non-zero
Gabidulin code $C \subseteq \F_{q^m}^k$ is $\mbox{minrk}(C):=
\min \{ \mbox{rk}(c) : c \in C, \ c \neq 0 \}$, and the \textbf{maximum
rank}
of any Gabidulin code $C$ is $\mbox{maxrk}(C) := \max \{ \mbox{rk}(c) : c
\in
C\}$.
Given
a $t$-dimensional   non-zero Gabidulin code $C \subseteq \F_{q^m}^k$ and an integer
$1 \le r \le t$, the \textbf{$r$-th generalized rank weight} of $C$ is  
$m_r(C):=
\min \{\dim_{\F_{q^m}}(V) : V \in \Lambda_q(k,m), \ \dim_{\F_{q^m}}(V \cap C)
\ge
r\}$.
\end{definition}

 In the literature researchers often call ``Gabidulin codes'' only the
 rank-metric codes obtained evaluating certain sets of linearized polynomials
 (see \cite{gabid}).
 For practical reasons we will not make this distinction here, and
 simply call 
 ``Gabidulin code'' any $\F_{q^m}$-subspace $C \subseteq \F_{q^m}^k$.

 In \cite{metrics3} Silva and Kschischang propose a coding scheme to
secure a network communication against an eavesdropper based on Gabidulin
codes. Generalized rank weights were introduced by
Kurihara,
 Matsumoto and Uyematsu in \cite{kmat} to measure the performance of a Gabidulin
code when
employed in the cited
scheme. The following theorem summarizes the main properties of generalized
rank weights established in \cite{kmat}. See also \cite{ducoat}.

\begin{theorem}[see \cite{kmat}] \label{propr2}
Let $C \subseteq \F_{q^m}^k$ be a non-zero
Gabidulin code of dimension $1 \le t \le k$ over $\F_{q^m}$. 
 The following hold.
 \begin{enumerate}
  \item $m_1(C)= \mbox{minrk}(C)$.
  \item $m_t(C) \le k$.
  \item For any $1 \le r \le t-1$ we have $m_r(C)< m_{r+1}(C)$.
  \item For any $1 \le r \le t$ we have $m_r(C) \le k-t+r$.
 \end{enumerate}
\end{theorem}

\section{Generalized Hamming weights and anticodes} \label{sss1}
In this section we characterize the generalized Hamming weights of a linear
code  in terms of the intersections of the code with optimal anticodes in the
Hamming metric. Recall that a matrix $M$ over a field is  in \textbf{row-reduced
echelon form} (abbreviated as ``\textbf{RRE form}'')
 if:
 \begin{itemize}
  \item each row of $M$ has more initial zeros than the previous rows,
  \item the first non-zero entry of any non-zero row of $M$ (called the
\textbf{pivot entry} of the row) equals $1$, and 
  it is also the only non-zero entry in its column.
 \end{itemize}
If $M$ is in RRE form, then the columns of $M$ that contain a pivot entry are
the \textbf{pivot columns} of $M$. 
Each pivot column contains only one non-zero entry, and such entry equals $1$.
It is well-known that any matrix can be put in row-reduced echelon form by
performing elementary operations
on the rows. Moreover, the row-reduced echelon form of a matrix is unique.
Therefore, given a field $\F$, an integer $s \ge 1$ and a subspace  $V \subseteq \F^s$ of
dimension $1 \le t \le s$, 
there exists a unique matrix $M \in \mbox{Mat}(t \times s, \F)$ in
row-reduced echelon form
such that $\mbox{rowsp}(M)=V$. We denote such matrix by $\mbox{RRE}(V)$.

Now if $C \subseteq \F_q^n$ is a non-zero linear code then the sum of the
rows of
$\mbox{RRE}(C)$ is a vector of Hamming weight at least $\dim(C)$. This
shows the following bound.

\begin{proposition} \label{bound1}
 Let $C \subseteq \F_q^n$ be a linear code. We have
$\dim_{\F_q}(C) \le \mbox{maxwt}(C)$.
\end{proposition}

\begin{definition} \label{antiH}
 A code $C \subseteq \F_q^n$ attaining the bound of Proposition \ref{bound1} is
 an \textbf{optimal linear anticode}. We denote the set of
optimal
linear anticodes in $\F_q^n$ by $\mA^H_q(n)$.
\end{definition}

One can  construct  simple optimal linear anticodes as
follows. 
 Let $S \subseteq [n]$ be any subset. The \textbf{free code} over $\F_q$ of length
$n$ \textbf{supported} on $S$ is
 $C_q(n,S):=\{ v \in \F_q^n : v_i=0 \mbox{ for all } i \in
[n] \setminus S\}$.
Clearly, any free code $C_q(n,S)$ has
$\dim_{\F_q}(C_q(n,S))=\mbox{maxwt}(C_q(n,S))=|S|$. Thus free
codes
are 
 optimal linear anticodes. Vice versa, we now show that for $q \ge 3$
all
optimal linear anticodes are free
codes.

\begin{lemma} \label{lll1}
 Assume $q \ge 3$. Let $t \ge 1$ be an integer, and let $c_1,...,c_t \in \F_q$
be not all zero.
There exist $a_1,...,a_t \in \F_q\setminus \{ 0 \}$ such that $\sum_{i=1}^t
a_ic_i \neq 0$.
\end{lemma}

\begin{proof}
 Choose $b_1,...,b_t \in \F_q\setminus \{ 0 \}$. If $\sum_{i=1}^t
b_ic_i \neq 0$ then take $a_i=b_i$ for $i \in [t]$.
Assume $\sum_{i=1}^t
b_ic_i = 0$. By hypothesis, there exists $j \in [t]$ such that $c_j \neq 0$. 
Let $b \in
\F_q \setminus \{0,1\}$. Define $a_j:=bb_j$, and $a_i:=b_i$ for $i \in [t]
\setminus \{ j\}$. Since $b \neq 0$ we have $a_i \neq 0$ for all $i
\in [t]$.
Moreover,
$$\sum_{i=1}^t a_ic_i = bb_jc_j+ \sum_{i \neq j} b_ic_i=b_jc_j+(b-1)b_jc_j +
\sum_{i \neq j} b_ic_i = \sum_{i=1}^t b_ic_i+(b-1)b_jc_j=(b-1)b_jc_j.$$
Since $b \neq 1$, $b_j \neq 0$ and $c_j \neq 0$ we have $(b-1)b_jc_j \neq 0$.
\end{proof}

\begin{proposition} \label{antifree}
 Assume $q \ge 3$. Let $C \subseteq \F_q^n$ be a linear code of dimension
$t$. Then $C \in \mA_q^H(n)$ if and only if $C=C_q(n,S)$ for some $S
\subseteq [n]$ with $|S|=t$.
\end{proposition}

\begin{proof}
The implication $(\Leftarrow)$ is clear. Let us prove 
$(\Rightarrow)$. If $t=0$ or $t=n$ then the result is trivial.
Assume $0 < t < n$.  
If $C$ is an optimal
anticode we have $t=\mbox{maxwt}(C)$. Let $M:=\mbox{RRE}(C)$. We will 
show that any non-pivot column of $M$ is zero. By contradiction, let $j \in
[n]$ be the index of a non-zero non-pivot column of $M$, and let
$c^j_1,...,c^j_t$ be the entries of such column. By Lemma \ref{lll1} there
exist
$a_1,...,a_t \in \F_q \setminus \{ 0\}$
with $\sum_{i=1}^t
a_ic^j_i \neq 0$. Denote by $M_1,...,M_t \in \F_q^n$ the rows of $M$. We
have that $\sum_{i=1}^t a_iM_i \in C$ has Hamming weight at least $t+1$, a
contradiction. It follows $c^j_i =0$ for all $i \in [t]$.  Hence we proved
$C \subseteq C_q(n,S)$, where $S \subseteq [n]$ is the set of pivot columns
of $M$. In particular, $|S|=t$, and so 
$C = C_q(n,S)$.
\end{proof}

Proposition \ref{antifree} allows us to characterize the generalized
Hamming weights of a linear code in terms of optimal
anticodes as follows.

\begin{theorem}\label{cara1}
 Assume $q \ge 3$. Let $C \subseteq \F_q^n$ be a non-zero linear code of
dimension
 $1 \le t \le n$. For
any integer $1
\le r \le t$ we have $d_r(C) = \min \{\dim_{\F_q}(A) : A \in \mA^H_q(n), \
\dim_{\F_q}(A \cap C) \ge
r\}$.
\end{theorem}

\begin{proof}
 Fix $1 \le r \le t$. Define $d'_r(C) := \min \{\dim_{\F_q}(A) : A \in \mA^H_q(n), \ \dim_{\F_q}(A
\cap C) \ge
r\}$. Let $A \in \mA^H_q(n)$ with $\dim_{\F_q}(A)=
d'_r(C)$ and $\dim_{\F_q}(A
\cap C) \ge
r$. By Proposition \ref{antifree}, $A=C_q(n,S)$ for some
$S
\subseteq [n]$ with $|S|=\dim_{\F_q}(A)$. Let $D$ be an
$r$-dimensional subspace of $A \cap C$. We have $\chi(D) \subseteq
\chi(A \cap C) \subseteq \chi(A)=\chi(C_q(n,S))=S$, and so
$|\chi(D)| \leq |S|=\dim_{\F_q}(A)$. This proves $d_r(C) \le
d'_r(C)$.  Let now $D \subseteq C$ with
$\dim_{\F_q}(D)=r$ and $|\chi(D)|=d_r(C)$. Define $A:=C_q(n,\chi(D))$. Since $A
\supseteq D$ and $D \subseteq C$, we have $\dim_{\F_q}(A \cap C) \ge \dim_{\F_q}(D \cap C)= \dim_{\F_q}(D)=r$.
Moreover, $\dim_{\F_q}(A)=|\chi(D)|=d_r(C)$, and so $d'_r(C) \le d_r(C)$.
\end{proof}

 Notice that Theorem \ref{cara1} and Proposition \ref{antifree} do not hold in
general
when $q=2$. Take e.g. $n=3$, and
 let $C$ be the linear code generated over $\F_2$ by $(1,0,1)$ and
$(0,1,1)$. We have $d_2(C)=|\chi(C)|=3$. On the other hand, $C$ is an optimal
linear anticode of maximum weight $2$, even if it is not of the form
$C_2(3,S)$ for some $S\subseteq [n]$ with $|S|=2$. Following the notation of the proof of Theorem
\ref{cara1} we have $d'_2(C)=\dim_{\F_q}(C)=2 \neq d_2(C)$.

\section{Generalized rank weights and anticodes} \label{sss2}

The aim of this section is to establish the analogue of Theorem \ref{cara1} for Gabidulin
codes and generalized rank weights. We start studying optimal anticodes in
the rank
metric, giving a bound on their dimension.

\begin{proposition} \label{bound2}
 Let $C \subseteq \F_{q^m}^k$ be a 
Gabidulin code. We have
$\dim_{\F_{q^m}}(C) \le \mbox{maxrk}(C)$.
\end{proposition}

\begin{proof}
 If $C=0$ the result is trivial. Assume $t:=\dim_{\F_{q^m}}(C) \ge 1$ and
let $M_1,...,M_t$ denote the rows of $M:=\mbox{RRE}(C) 
 \in \mbox{Mat}(t \times k, \F_{q^m})$. Let
$\alpha_1,...,\alpha_t \in \F_{q^m}$ be independent over $\F_q$. Then
$\sum_{i=1}^t \alpha_iM_i \in C$ has
$\alpha_1,...,\alpha_t$ among its components. In particular,
$\mbox{rk}(\sum_{i=1}^t \alpha_iM_i) \ge t$.
\end{proof}

\begin{definition} \label{antiR}
 A code $C \subseteq \F_{q^m}^k$ attaining the bound of Proposition
\ref{bound2} is  an \textbf{optimal Gabidulin anticode}.
We denote the set of optimal Gabidulin anticodes in $\F_{q^m}^k$ by
$\mA^G_q(k,m)$.
\end{definition}

We now present a series of preliminary results relating Frobenius-closed spaces,
matrices in RRE form, and optimal anticodes.

\begin{theorem}[\cite{gp}, Theorem 1] \label{defsufq}
Let $V \subseteq \F_{q^m}^k$ be an $\F_{q^m}$-subspace. Then $V \in
\Lambda_q(k,m)$
if and only if $V$ has a basis made of vectors with entries in
$\F_q$ (in short, \textbf{defined} over $\F_q$).
\end{theorem}

Combining Theorem \ref{defsufq} with the uniqueness of the RRE form 
we obtain the following useful criterion to test whether a space is
Frobenius-closed or not.
The result also provides  an efficient way to
compute
a basis defined over $\F_q$ of a Frobenius-closed space $V \subseteq
\F_{q^m}^k$, as we will show in an example.

\begin{corollary} \label{proprbase}
 Let $V \subseteq \F_{q^m}^k$ be a non-zero subspace.
Then $V \in \Lambda_q(k,m)$ if and only if $\mbox{RRE}(V)$ is a matrix with
entries in $\F_q$.
\end{corollary}

\begin{example}
 Let $q=2$ and $k=m=4$. Write $\F_{2^4}=\F_{2}[\xi]$, where $\xi$ satisfies
 $\xi^4+\xi+1 =0$. Let $V \subseteq \F_{2^4}^4$ be the space  
generated by the vectors
 $v_1:=(\xi, \xi^2,\xi^5,\xi)$ and $v_2:=(\xi^2,\xi^4,\xi^{10},\xi^2)$, and let
$M$ denote the matrix having $v_1$ and $v_2$ as rows. 
The RRE form of $M$ is $$\begin{bmatrix}
    1 & 0 & 1 & 1 \\ 0 & 1 & 1 & 0
   \end{bmatrix}.$$ Therefore 
 $V$ is Frobenius-closed, and $\{
(1,0,1,1), (0,1,1,0)\}$ is a basis of $V$ defined over $\F_2$.
\end{example}

We will need the following preliminary lemma.

\begin{lemma}\label{tecn}
 Let $H \subseteq \F_{q^m}$ be an $\F_q$-subspace of dimension $h$ over $\F_q$,
 with $1 \le h \le m-2$. Let $x \in \F_{q^m}\setminus H$, and $y \in \F_{q^m}$.
 There exists $\alpha \in \F_{q^m} \setminus H$ such that $x+\alpha y \notin H
\oplus \langle \alpha \rangle$, where $\langle \alpha \rangle \subseteq
\F_{q^m}$ denotes the space generated by $\alpha$ over $\F_q$.
\end{lemma}

\begin{proof}
 Define the sets 
$U:=\{ a \in \F_q : a \neq y\}$ and $\mathcal{U}:= \{ \alpha \in
\F_{q^m} : \exists \ v \in H, \ a \in U \ \mbox{ with } \ \alpha=(v-x)/(y-a)\}$.
We claim that $x+ \alpha y \in H \oplus \langle \alpha \rangle$ if and only if
$\alpha \in \mU$. Indeed, if $\alpha \in \mU$ then $\alpha=(v-x)/(y-a)$ for some
$v \in H$ and $a \in U \subseteq \F_q$. Hence $\alpha(y-a)=v-x$, and so
$x+\alpha y = v+a\alpha \in H \oplus \langle \alpha \rangle$. Vice versa, if
$x+\alpha y \in H \oplus \langle \alpha \rangle$ then there exist $v \in H$ and
$a \in \F_q$ with $x+\alpha y=v+a\alpha$. If $a=y$ then $x=v \in H$, a
contradiction. It follows  $a \in U$, and $\alpha=(v-x)/(y-a)$.

We clearly have $|\mU| \leq |H| \cdot |U| \le q^hq=q^{h+1}$. Hence
$|\F_{q^m} \setminus \mU| \ge q^m-q^{h+1}$. Since $m-h \ge 2$ by hypothesis, we
have $ q^{m-h}-q \ge q^2-q>1$.
Multiplying both members of this inequality  by $q^h$ we obtain $q^m-q^{h+1}
>q^h$. Hence we have $|\F_{q^m} \setminus \mU| \ge q^m-q^{h+1} > q^h$. Since
$|H|=q^h$, there exists $\alpha \in (\F_{q^m} \setminus \mU) \setminus H$.
Since $\alpha \notin \mU$ we have $x+ \alpha y \notin H \oplus \langle \alpha
\rangle$ by the claim.
\end{proof}

 Given a matrix $M$ with $t$ rows $M_1,...,M_t$ and a permutation 
$\pi:[t] \to [t]$, we denote by $\pi(M)$ the matrix whose rows are
$M_{\pi(1)},...,M_{\pi(t)}$.
 A matrix $M$  is \textbf{almost}
in RRE form if $\pi(M)$ is in RRE form
for some
permutation $\pi$.

\begin{proposition} \label{induz}
Let $1 \le t <k$ be an integer, and let $M \in \mbox{Mat}(t \times k,
\F_{q^m})$
be a full-rank matrix almost in RRE form with rows
$M_1,...,M_t$. If $M_1$
has at least one
entry in $\F_{q^m} \setminus \F_q$, then there exist
$\F_q$-linearly independent elements $\alpha_1,...,\alpha_t \in \F_{q^m}$ such
that $\mbox{rk}(\sum_{i=1}^t \alpha_i M_i) \ge t+1$. 
\end{proposition}

\begin{proof}
By induction on $t$.
If $t=1$ then $M$ has only one row, $M_1 \in \F_{q^m}^k$. Such row has $1$ and an element $M_{1j} \notin \F_q$
among
its
entries. In particular, it has rank $\ge 2$, and we can take $\alpha_1:=1$ to
conclude the proof.
Assume that the result holds for all non-negative integers smaller than $t$.
Denote by $M' \in \mbox{Mat}(t-1,k, \F_{q^m})$ the matrix obtained from $M$
deleting the last row. Clearly, $M'$ has full-rank
and it is almost in RRE form.
By induction hypothesis there are $\alpha_1,...,\alpha_{t-1} \in \F_{q^m}$
independent over $\F_q$ with $\mbox{rk}(\sum_{i=1}^{t-1} \alpha_i M_i)
\ge t$. Since the vector
$\sum_{i=1}^{t-1} \alpha_i M_i$ has $\alpha_1,...,\alpha_{t-1}$
among its components, there exists $j\in [k]$ with $\sum_{i=1}^{t-1}
\alpha_i M_{ij} \notin \langle  \alpha_1,...,\alpha_{t-1} \rangle$.
 Lemma \ref{tecn} with $H=\langle  \alpha_1,...,\alpha_{t-1} \rangle$,
$x= \sum_{i=1}^{t-1}
\alpha_i M_{ij}$, $y=M_{tj}$ gives an element
$\alpha_t \in \F_{q^m} \setminus \langle  \alpha_1,...,\alpha_{t-1} \rangle$
with $\sum_{i=1}^{t-1}
\alpha_i M_{ij} + \alpha_t M_{tj}= \sum_{i=1}^t \alpha_i M_{ij} \notin  \langle
 \alpha_1,...,\alpha_{t} \rangle$. Thus
 $\sum_{i=1}^t \alpha_i M_i$ has rank $\ge t+1$.
\end{proof}

The following theorem shows that Frobenius-closed spaces coincide with optimal Gabidulin anticodes.
In particular, it shows that the algebraic condition of being
 Frobenius-closed may be regarded as a metric
condition.

\begin{theorem} \label{casino}
We have $\Lambda_q(k,m)= \mA^G_q(k,m)$.
\end{theorem}

\begin{proof}
 Let
$V \in \Lambda_q(k,m)$. Denote by $t$ the dimension of $V$ over $\F_{q^m}$. If
$t=0$ then clearly $V \in \mA^G_q(k,m)$. Now assume $1 \le t \le k$. By Theorem
\ref{defsufq} there exists a basis $\{v_1,...,v_t\}$ of $V$ defined over $\F_q$.
Take any $v \in V$. There exist $\alpha_1,...,\alpha_t \in \F_{q^m}$ with
$v= \sum_{i=1}^t \alpha_i v_i$. The space generated over $\F_q$ by the
entries of $v$ is contained in $\mbox{Span}_{\F_q} \{
\alpha_1,...,\alpha_t\}$. In particular $\mbox{rk}(v) \le t$. Since
$v \in V$ is arbitrary, this
proves
$\mbox{maxrk}(V) \le t$. By Proposition \ref{bound2} we have $\mbox{maxrk}(V) =
t=\dim_{\F_{q^m}}(V)$, and so $V
\in \mA^G_q(k,m)$. Now we prove $\mA^G_q(k,m) \subseteq \Lambda_q(k,m)$. Let $A
\in \mA^G_q(k,m)$, and denote by $t$ the dimension of $A$ over $\F_{q^m}$. If
$t=0$
or $t=k$ then $A \in \Lambda_q(k,m)$. Assume $1 \le t <k$, and set
$M:=\mbox{RRE}(A)$. By Corollary \ref{proprbase} it suffices to show that $M$
has entries in $\F_q$. By contradiction, assume that $M$ has one entry, say
$M_{ij}$, in $\F_{q^m} \setminus \F_q$. Exchanging the first and the
$i$-th row of $M$ we obtain a matrix, say $N$, almost in RRE form
 such that $\mbox{rowsp}_{\F_{q^m}}(N)=\mbox{rowsp}_{\F_{q^m}}(M)=A$. By 
 Proposition \ref{induz} there exists $v \in \mbox{rowsp}_{\F_{q^m}}(N)=A$ with
 $\mbox{rk}(v) \ge t+1$, and this contradicts the fact that $A$ is an optimal
anticode of dimension $t$.
\end{proof}

We can now state the main result of this section, characterizing 
generalized rank weights in terms of optimal Gabidulin anticodes. The result
follows from
Definition \ref{rgw} and
Theorem \ref{casino}, and it may be regarded as the analogue of Theorem
\ref{cara1}
for Gabidulin codes.

\begin{corollary}\label{cara2}
Let $C \subseteq \F_{q^m}^n$ be a non-zero
Gabidulin code
of dimension $1 \le t \le k$ over $\F_{q^m}$. For all $1
\le r \le t$ we have $m_r(C) = \min \{\dim_{\F_{q^m}}(A) : A \in \mA^G_q(k,m),
\ \dim_{\F_{q^m}}(A \cap C) \ge r\}$.
\end{corollary}

\section{An algebraic invariant for Delsarte codes} \label{sss4}

In \cite{del1} Delsarte defines rank-metric codes as
linear spaces of matrices over a finite field.
In this section we briefly recall the basic definitions, and
propose a new algebraic invariant for Delsarte codes in analogy with the
generalized Hamming weights for linear codes and with the
generalized rank weights for Gabidulin codes.

\begin{definition} \label{variedefgab}
 A \textbf{Delsarte code}  
is an $\F_q$-subspace 
 $\mC \subseteq \mbox{Mat}(k \times m, \F_q)$.
  The \textbf{minimum rank} of a non-zero Delsarte code $\mC$ is 
 $\mbox{minrk}(\mC):= \min \{ \mbox{rk}(M) : M \in \mC, \ \mbox{rk}(M) >0\}$.
The \textbf{maximum rank} of any Delsarte code $\mC$
 is  $\mbox{maxrk}(\mC):= \max \{ \mbox{rk}(M) : M \in \mC\}$.
 \end{definition}

In analogy with Proposition \ref{bound1} and Proposition \ref{bound2} we have
the following bound.

\begin{proposition}[see \cite{io}, Proposition 47] \label{anticode}
 Let $\mC \subseteq \mbox{Mat}(k\times m,
\F_q)$ be a Delsarte code. We have
 $\dim_{\F_q}(\mC) \le m \cdot \mbox{maxrk}(\mC)$.
\end{proposition}

\begin{definition} \label{defanticode}
 A code $\mC \subseteq \mbox{Mat}(k\times m, \F_q)$ attaining the
bound of 
Proposition \ref{anticode}
 is  a \textbf{Delsarte optimal anticode}. We denote by $\mA^D_q(k,m)$
the set of Delsarte optimal anticodes
 in the space $\mbox{Mat}(k \times m, \F_q)$.
\end{definition}

Inspired by Theorem \ref{cara1} and
Corollary \ref{cara2}, we propose the following definition.

\begin{definition} \label{nsdef}
 Let $\mC \subseteq \mbox{Mat}(k \times m,
\F_q)$ be a non-zero Delsarte code of dimension $1 \le t \le km$.
For
 $1 \le r \le t$, the \textbf{$r$-th Delsarte generalized weight} of $\mC$ is
 $$a_r(\mC):=\frac{1}{m} \min \{\dim_{\F_q}(\mA) : \mA \in \mA^D_q(k,m), \
\dim_{\F_q}(\mA \cap
\mC) \ge r\}.$$
\end{definition}

 By Definition \ref{defanticode}, the dimension over $\F_q$ of any anticode $A
\in \mA^D_q(k,m)$ is
a multiple of $m$. Therefore Delsarte generalized weights are positive integers.

Before describing the properties of Delsarte generalized weights 
we show how our invariant relates 
to the generalized rank weights for Gabidulin codes of \cite{kmat}. Writing the
components of a
vector $v \in \F_{q^m}^k$ over a basis $\mG$
of $\F_{q^m}$ over $\F_q$ one can naturally associate to a Gabidulin code a
Delsarte code with the same metric properties.

\begin{definition}\label{assoc}
 Let $\mG= \{ \gamma_1,...,\gamma_m\}$ be a basis
of $\F_{q^m}$ over $\F_q$. The matrix \textbf{associated} to a vector
$v \in \F_{q^m}^k$ with respect to $\mG$ is the matrix 
$M_{\mG}(v) \in \mbox{Mat}(k \times m,\F_q)$ defined by
$v_i= \sum_{j=1}^m
M_{\mG}(v)_{ij}\gamma_j$ for all $i \in [k]$. The Delsarte code
\textbf{associated} to a Gabidulin code $C \subseteq \F_{q^m}^k$ with respect to
the basis $\mG$
 is $\mC_{\mG}(C):=\{ M_{\mG}(c) : c \in C\}
\subseteq
 \mbox{Mat}(k \times m, \F_q)$.
\end{definition}

 If $C \subseteq \F_{q^m}^k$ is a Gabidulin code, then for any basis $\mG$ of
$\F_{q^m}$ over $\F_q$ we have that $\mC_{\mG}(C) \subseteq
\mbox{Mat}(k\times m, \F_q)$
is a Delsarte rank-metric code with
 $\dim_{\F_q}
\mC_{\mG}(C)=m \cdot \dim_{\F_{q^m}}(C)$.
 Moreover, $\mbox{maxrk}(C)=\mbox{maxrk}(\mC_{\mG}(C))$, and if $C \neq 0$   
 we have $\mbox{minrk}(C)= \mbox{minrk}(\mC_{\mG}(C))$.

 Recall that, by definition, Gabidulin codes are $\F_{q^m}$-linear, while 
Delsarte codes are $\F_q$-linear. Therefore Gabidulin codes can be regarded as a
proper subset
of Delsarte codes. More precisely, a Delsarte code $\mC \subseteq
\mbox{Mat}(k\times m, \F_q)$ is Gabidulin if and only if it is $\F_{q^m}$-linear
under some
isomorphism $\F_{q^m}^k \to \mbox{Mat}(k\times m, \F_q)$ of the form $v \mapsto
M_{\mG}(v)$.

Since Delsarte codes generalize Gabidulin codes,
one would
expect
that Delsarte generalized weights refine, as an invariant, generalized
rank
weights. In the remainder of this section we show precisely this fact. We
start 
introducing some rank-preserving transformations.

\begin{notation} \label{r}
Given a Gabidulin code $C \subseteq \F_{q^m}^k$, a Delsarte code $\mC \subseteq
\mbox{Mat}(k \times m, \F_q)$ and matrices 
 $A \in \mbox{Mat}(k \times k, \F_q)$, $B \in \mbox{Mat}(m \times m, \F_q)$, 
define:
$$CA:= \{ c A : c \in C\}, \ \ \ \ A\mC:= \{ AM : M \in \mC\}, \ \ \ \ \mC B:=
\{ MB : M \in \mC\}.$$
 It is easy to see that if $A$ and $B$ are invertible matrices, then these 
multiplication maps are
 rank-preserving isomorphisms of Gabidulin and Delsarte codes. In particular,
they preserve optimal anticodes
 in the respective metrics,  generalized rank weights and Delsarte generalized
weights.
If $k=m$,  define the \textbf{transpose} of a Delsarte code $\mC
\subseteq \mbox{Mat}(k \times k, \F_q)$  by
 $\mC^t:= \{ M^t : M \in \mC\} \subseteq \mbox{Mat}(k \times k, \F_q)$.
It is easy to check that $\mC$ and $\mC^t$ have the same Delsarte
generalized weights.
\end{notation}

One can construct a simple family of Delsarte optimal anticodes as
follows. 
 Let $0 \le R \le k$ be an integer. The \textbf{standard
optimal anticode} $\mS_q(k,m,R)$ 
of maximum rank $R$
 is the vector space of $k \times m$ matrices over $\F_q$ whose last $k-R$ rows
equal zero. 
 The following result shows that, up to the rank-preserving transformations
introduced in Notation 
\ref{r}, 
all Delsarte optimal anticodes
are standard optimal anticodes.

\begin{theorem}[\cite{pazzis}, Theorem 4 and Theorem 6] \label{paz}
 Let $1 \le R \le k \le m$ be integers, and let $\mA \in \mA^D_q(k,m)$ with
$\mbox{maxrk}(\mA)=R$.
\begin{enumerate}
 \item If $k<m$ then there exist invertible matrices $A \in \mbox{Mat}(k \times
k,\F_q)$, $B \in
\mbox{Mat}(m \times m,\F_q)$ such that
 $A \mA B=\mS_q(k,m,R)$.
 \item If $k=m$ then  there exist invertible matrices $A,B \in \mbox{Mat}(k
\times k,\F_q)$ such that either
 $A \mA B=\mS_q(k,k,R)$, or $A \mA B=\mS_q(k,k,R)^t$.
\end{enumerate}
\end{theorem}

\begin{proof}
 If $R=0$ or $R=k$ then the result is trivial. Assume $1 \le R \le k-1$. If
$k<m$ the result
 follows (up to a transposition) from \cite{pazzis}, Theorem 6(a). If $k=m$ and
$R>1$ then
 apply  \cite{pazzis}, Theorem 4(a). Finally, if $k=m$ and
 $R=1$  the result follows from \cite{pazzis}, Theorem 4(b).
\end{proof}

We will also need the following linear algebra result, whose
proof is left to the reader.

\begin{lemma} \label{tett}
 Let $C\subseteq \F_{q^m}^k$ be a Gabidulin code. The following hold.
\begin{enumerate}
 \item If $A \in
\mbox{Mat}(k\times k, \F_q)$
 is an invertible matrix, then for any basis $\mG$ of $\F_{q^m}$ over $\F_q$
we have
 $\mC_\mG(CA^t)=A \mC_\mG(C)$.
 In particular, $\mC_\mG(C)$ and $\mC_\mG(CA^t)$ have the same Delsarte
generalized weights.\label{te1}

\item Let $\mG=
\{\gamma_1,...,\gamma_m \}$, $\mF:= \{ \varphi_1,...,\varphi_m\}$
 be bases of $\F_{q^m}$ over $\F_q$, and let $B \in \mbox{Mat}(m \times m,
\F_q)$ denote
the invertible matrix
 defined by
 $\gamma_j =\sum_{s=1}^m B_{js}\varphi_s$ for all $j \in [m]$. We have
 $\mC_\mF(C) = \mC_\mG(C)B$.
 In particular, if $C \neq 0$ then the Delsarte generalized weights of
$\mC_\mG(C)$ do not depend on the choice of the basis $\mG$.\label{te2}

\item Let $D \subseteq \F_{q^m}^k$ be another Gabidulin code, and let $\mG$ be a
basis of
$\F_{q^m}$ over $\F_q$. We have $\mC_\mG(C \cap D)= \mC_\mG(C) \cap
\mC_\mG(D)$.\label{te3}
\end{enumerate}
\end{lemma}

We can now prove that Delsarte generalized weights refine, as an algebraic
invariant,
generalized rank weights. 

\begin{theorem} \label{finer}
 Let $C \subseteq \F_{q^m}^k$ be a non-zero
Gabidulin code of dimension $1 \le t \le k$. For any basis
 $\mG$ of $\F_{q^m}$ over $\F_q$ and for any integers $1 \le r \le t$ and $0
\le
\varepsilon \le m-1$ we have $m_r(C)= a_{rm-\varepsilon}(\mC_\mG(C))$.
 In particular, the Delsarte generalized weights of a Delsarte $\mC$ code
 arising from a Gabidulin code are fully determined by a suitable subset of
them.
\end{theorem}

\begin{proof}
 Fix $1 \le r \le t$ and $0 \le \varepsilon \le m-1.$
 Let $\overline{A} \in \mA^G_q(k,m)$ with $\dim_{\F_{q^m}}(\overline{A})=m_r(C)$
and $\dim_{\F_{q^m}}(\overline{A}\cap C) \ge r$. We have $\mC_\mG(\overline{A})
 \in \mA_q^D(k,m)$ and $\dim_{\F_q}(\mC_\mG(\overline{A}))=m
\cdot \dim_{\F_{q^m}}(\overline{A})=m \cdot m_r(C)$. 
 By Lemma \ref{tett}(\ref{te3}), $\mC_\mG(\overline{A}) \cap \mC_\mG(C)=
\mC_\mG(\overline{A} \cap C)$. Hence
  we have $$\dim_{\F_q}(\mC_\mG(\overline{A}) \cap
\mC_\mG(C)) = \dim_{\F_q}(\mC_\mG(\overline{A} \cap  C)) \ge rm \ge
rm-\varepsilon.$$ It 
 follows $a_{rm-\varepsilon}(\mC_\mG(C)) \le  m_r(C)$.
 
 Now we prove $m_r(C) \le a_{rm-\varepsilon}(\mC_\mG(C))$. Define $\mC:=
\mC_\mG(C)$ to simplify the notation.
Let $\overline{\mA} \in \mA_q^D(k,m)$ with $\dim_{\F_q}(\overline{\mA}\cap \mC)
\ge rm-\varepsilon$ and 
 $a_{rm-\varepsilon}(\mC)=1/m \cdot \dim_{\F_q}(\overline{\mA})$. By Definition
\ref{defanticode}, $\dim_{\F_q}(\overline{\mA})=mR$, where
$R=\mbox{maxrk}(\overline{\mA})$.
 Hence we need to prove $m_r(C) \le R$. 
  By Theorem \ref{paz} there exist
invertible matrices $A \in \mbox{Mat}(k \times k, \F_q)$ and 
 $B \in \mbox{Mat}(m \times m, \F_q)$ such
 that either $A \overline{\mA} B= \mS_q(k,m,R)$,
 or $k=m$ and $A \overline{\mA} B= \mS_q(k,k,R)^t$.
 By Remark \ref{r} (replacing if necessary $\mC$ with $\mC^\perp$,
$\overline{\mA}$ with
 $\overline{\mA}^\perp$, $A$ with $B^t$ and $B$ with
 $A^t$) without loss of generality we may assume to
 be in the former case.
 Let $\mG= \{
\gamma_1,...,\gamma_m\}$, and for $i \in [m]$ define
 $\varphi_i:= \sum_{j=1}^m B^{-1}_{ij} \gamma_j$. It is clear that $\mF:= \{
\varphi_1,...,\varphi_m\}$ is a basis of $\F_{q^m}$ over
 $\F_q$. 
  Define  the optimal Gabidulin anticode $V:= \{ v \in  \F_{q^m}^k : v_i
=0 \mbox{ for } i >R\} \subseteq \F_{q^m}^k$.
Using Definition \ref{assoc} one can  check that $\mC_\mF(V)=
\mS_q(k,m,R)=A\overline{\mA}B$.
Since $V$ is an
 optimal Gabidulin anticode of dimension $R$ over $\F_{q^m}$, by Remark \ref{r} 
 $V (A^t)^{-1}$ is an optimal Gabidulin anticode of
 dimension $R$ as well. Hence by Corollary \ref{cara2} it suffices to prove
$\dim_{\F_{q^m}}(V (A^t)^{-1} \cap C) \ge r$.
  By  Lemma \ref{tett}(\ref{te3}) we have
 \begin{eqnarray*} \label{ee1}
  \dim_{\F_{q^m}}(V(A^t)^{-1} \cap C) &=& \dim_{\F_{q^m}}(V(A^t)^{-1}A^t \cap
CA^t) \\
  &=& \dim_{\F_{q^m}}(V \cap C A^t) \\
  &=& \frac{1}{m}  \dim_{\F_q}(\mC_\mF(V \cap CA^t)) \\
  &=& \frac{1}{m}  \dim_{\F_q}(\mC_\mF(V) \cap \mC_\mF(CA^t)).
 \end{eqnarray*}
By Lemma \ref{tett}, parts \ref{te1} and \ref{te2}, we have
$\mC_\mF(CA^t)=A\mC_\mF(C)=A \mC_\mG(C)B=A \mC B$.
It follows $$\mC_\mF(V) \cap \mC_\mF(CA^t)=A \overline{\mA}B \cap A \mC
B=A(\overline{\mA} \cap \mC)B.$$ Since $\dim_{\F_q}(A(\overline{\mA} \cap
\mC)B)=\dim_{\F_q}(\overline{\mA} \cap \mC)$, we have
$$\frac{1}{m}  \dim_{\F_q}(\mC_\mF(V) \cap \mC_\mF(CA^t))=\frac{1}{m}
\dim_{\F_q}(\overline{\mA} \cap \mC) \ge \frac{1}{m} (rm-\varepsilon).$$
It follows $\dim_{\F_{q^m}}(V(A^t)^{-1} \cap C) \ge \lceil (rm-\varepsilon)/m
\rceil=r$,
as claimed.
\end{proof}

 It is not true in general that for a Delsarte code $\mC \subseteq
\mbox{Mat}(k\times m, \F_q)$ 
 of dimension $t$ we have $a_{im}(\mC)=a_{im-\varepsilon}(\mC)$ for all $i \ge
1$ and $1 \le \varepsilon \le m-1$
 with $1 \le im-\varepsilon \le t$. 
 For example, one can produce codes $\mC \subseteq \mbox{Mat}(3 \times 3, \F_2)$
 of dimension $6$ having the Delsarte generalized weights given 
 in Table \ref{tabella1}.
 The examples reflect the fact that not all Delsarte codes $\mC$ arise from a
Gabidulin code, even when
 $\dim_{\F_q}(\mC) \equiv 0 \mod m$.

 \begin{table}[h!]
   \centering
  \begin{tabular}{l|c|c|c|c|c|c|}
   & $a_1(\mC)$ & $a_2(\mC)$ & $a_3(\mC)$ & $a_4(\mC)$ & $a_5(\mC)$ & $a_6(\mC)$
\\ 
  \hline 
  \hline
  Code \#1 & 1 & 1 & 1 & 2 & 2 & 3 \\
  \hline
  Code \#2 & 1 & 1 & 2 & 2 & 2 & 3 \\
  \hline
  Code \#3 & 1 & 1 & 1 & 2 & 3 & 3 \\
  \hline
  Code \#4 & 1 & 1 & 2 & 2 & 3 & 3 \\
  \hline
  Code \#5 & 1 & 1 & 2 & 3 & 2 & 3 \\
  \hline
  Code \#6 & 1 & 2 & 2 & 2 & 3& 3   \\
  \hline
 \end{tabular}
 \caption{Delsarte generalized weights of six different codes. Each
line corresponds to a code.}
 \label{tabella1}
  \end{table}

\section{Properties of Delsarte generalized weights} \label{sss5}

In this section we establish the analogue of Theorem \ref{propr1} and
Theorem
\ref{propr2} for Delsarte codes and Delsarte generalized weights, and
characterize optimal Delsarte codes and anticodes in terms of their
Delsarte generalized
weights.

Recall that, by Theorem 5.4 of \cite{del1}, for any non-zero Delsarte code $\mC
\subseteq \mbox{Mat}(k \times m, \F_q)$ we have 
$\dim_{\F_q} (\mC) \le m (k-\mbox{minrk}(\mC)+1)$. The code $\mC$ is 
\textbf{optimal} (or \textbf{MRD})
if its parameters attain the bound.

\begin{lemma}\label{sottoa}
Let $\mA \in \mA^D_q(k,m)$ with
$\mbox{maxrk}(\mA) \ge 1$. There
 exists $\mA' \in \mA^D_q(k,m)$ with $\mA' \subseteq \mA$ and
$\dim_{\F_q}(\mA')=\dim_{\F_q}(\mA)-m$.
\end{lemma}

\begin{proof}
 Let $R:=\mbox{maxrk}(\mA)$. By Theorem \ref{paz} there exist invertible
matrices $A$ and $B$ over $\F_q$ of size $k \times k$ and $m \times m$,
respectively,
 such that either $A \mA B = \mS_q(k,m,R)$, 
 or $k=m$ and $A \mA B = \mS_q(k,k,R)^t$. In the former case
 set $\mA':= A^{-1} \mS_q(k,m,R-1)B^{-1} \subseteq \mA$, and in the latter
case
 set $\mA':= A^{-1} \mS_q(k,k,R-1)^tB^{-1} \subseteq \mA$.
 \end{proof}

\begin{theorem} \label{propr3}
 Let $\mC \subseteq \mbox{Mat}(k\times m,
\F_q)$ be a non-zero Delsarte code of dimension $1 \le t \le km$.
The
following hold.

\begin{enumerate}
 \item $a_1(\mC)=\mbox{minrk}(\mC)$. \label{propr3_1}
 \item $a_{t}(\mC) \le k$. \label{propr3_2}
 \item For any $1 \le r \le t-1$ we have $a_r(\mC) \le a_{r+1}(\mC)$.
\label{propr3_3}
 \item For any $1 \le r \le t-m$ we have $a_r(\mC)<a_{r+m}(\mC)$.
\label{propr3_4}
 \item For any $1 \le r \le t$ we have $a_r(\mC) \le k- \lfloor (t-r)/m
\rfloor$. \label{propr3_5}
\item For any $1 \le r \le t$ we have $a_r(\mC) \ge \lceil r/m \rceil$.
\label{propr3_6}
\end{enumerate}
\end{theorem}

\begin{proof}
 We will prove the six properties separately.
 \begin{enumerate}
  \item Let $M \in \mC$ with $d:=\mbox{rk}(M)=\mbox{minrk}(\mC) \ge 1$. There
are invertible
matrices $A$ and $B$ over $\F_q$ of size $k\times k$ and $m\times m$,
respectively, such that $AMB$ is the matrix whose first $d$ diagonal entries
are ones and whose other entries equal zero. Clearly, $AMB \in \mS_q(k,m,d)$.
Set $\mA:= A^{-1} \mS_q(k,m,d) B^{-1}$. By Notation \ref{r}, $\mA$ is an
optimal Delsarte anticode of dimension $md$ such that $M \in
\mC \cap \mA$. In particular $\dim_{\F_q} (\mC \cap \mA) \ge 1$, and so
 $a_1(\mC) \le d$. Since $\mC$ has minimum rank $d$, it is clear that
$a_1(\mC) \ge d$.

\item Any anticode $\mA \in \mA_q^D(k,m)$ has dimension at most $km$.

\item Any anticode $\mA \in \mA_q^D(k,m)$ with $\dim_{\F_q}(\mA \cap \mC)
\ge r+1$ satisfies $\dim_{\F_q}(\mA \cap \mC)
\ge r$.

\item Let $\mA \in \mA^D_q(k,m)$ with $\dim_{\F_q}(\mA \cap \mC) \ge r+m$ and 
$\dim_{\F_q}(\mA)=m \cdot a_{r+m}(\mC)$. By Lemma \ref{sottoa} there exists
an optimal anticode $\mA' \subseteq \mA$ with
$\dim_{\F_q}(\mA')=\dim_{\F_q}(\mA)-m$.
It suffices to prove $\dim_{\F_q}(\mA' \cap \mC) \ge r$. Since $\mA' \subseteq
\mA$, we have $\mA' \cap \mC = \mA' \cap (\mA \cap \mC)$. Hence
$\dim_{\F_q}(\mA' \cap \mC)=\dim_{\F_q}(\mA' \cap (\mA \cap \mC))=
\dim_{\F_q}(\mA') + \dim_{\F_q}(\mA \cap \mC) -\dim_{\F_q}(\mA' + (\mA \cap
\mC))$.
Since $\mA' + (\mA \cap \mC) \subseteq \mA$, we have $\dim_{\F_q}(\mA' + (\mA
\cap \mC)) \le \dim_{\F_q}(\mA)$. As a consequence,
$\dim_{\F_q}(\mA' \cap \mC) \ge \dim_{\F_q}(\mA') + \dim_{\F_q}(\mA \cap \mC) 
- \dim_{\F_q}(\mA)= \dim_{\F_q}(\mA \cap \mC)-m \ge r$.

\item Define $h:= \lfloor (t-r)/m \rfloor$. By part $(2)$ and $(4)$ we find a
strictly increasing sequence of integers $a_r(\mC) < a_{r+m}(\mC) < \cdots <
a_{r+hm}(\mC) \le k$.
It follows $k \ge a_r+h$, i.e., $a_r \le k-h$.

\item If $\mA \in \mA_q^D(k,m)$ satisfies $\dim_{\F_q}(\mA \cap \mC) \ge r$
then, in particular,
$\dim_{\F_q}(\mA) \ge r$. Hence we have $a_r(\mC) \ge r/m$, i.e.,
$a_r(\mC) \ge
\lceil r/m \rceil$. \qedhere
 \end{enumerate}
\end{proof}

We can now show that Delsarte generalized weights completely characterize 
optimal codes and anticodes.

\begin{corollary} \label{optcode}
 Let $\mC \subseteq
\mbox{Mat}(k\times m, \F_q)$ be a non-zero Delsarte code with
$\dim_{\F_q}(\mC)=mR$.
The following are equivalent.
\begin{enumerate}
 \item $\mC$ is a Delsarte optimal code,
 \item $a_1(\mC)=k-R+1$,
 \item for all $r \in [mR]$ we have $a_r(\mC)=k-R+ \lceil r/m \rceil$.
\end{enumerate}
In particular, the Delsarte generalized
weights of a Delsarte optimal code $\mC \subseteq \mbox{Mat}(k\times m, \F_q)$
only
depend on $k$, $m$ and $\mbox{minrk}(\mC)$.
\end{corollary}

\begin{proof}
 By Theorem \ref{propr3}, $(1)$ and $(2)$ are equivalent.
 Assume $a_1(\mC)=k-R+1$. By Theorem \ref{propr3}, for all
$r \in [mR]$ we have
 $a_r(\mC) \le k - \lfloor (mR-r)/m \rfloor = k-R+ \lceil r/m \rceil$.
 Assume by contradiction that there exists $r \in [mR]$ with
 $a_r(\mC) < k-R+ \lceil r/m \rceil$.
 Define the non-negative integer $s:= \max \{ i \in \N : r-im \ge 1\}$.  We have
 $1 \le r-sm \le m$. In particular, $s \ge (r-m)/m=r/m-1$. Hence $s \ge \lceil
r/m \rceil-1$.
 By Theorem \ref{propr3} we have
\begin{eqnarray*}
 k-R+1 = a_1(\mC) &\le& a_{1+sm}(\mC)-s \\
 &\le& a_r(\mC)-s \\
 &<& k-R+\lceil r/m \rceil-s \\
 &\le& k-R+\lceil r/m \rceil- \lceil r/m \rceil +1 \\
 &=&k-R+1,
\end{eqnarray*}
a contradiction. Therefore we have $a_r(\mC) =k-R+ \lceil r/m \rceil$ for all $r
\in [mR]$. This proves
$(2) \Rightarrow (3)$. Finally, it is clear that $(3)$ implies $(2)$.
\end{proof}

\begin{corollary} \label{optanti}
 Let $\mC \subseteq
\mbox{Mat}(k\times m, \F_q)$ be a Delsarte code with $\dim_{\F_q}(\mC)=mR$.
The following are equivalent.
\begin{enumerate}
 \item $\mC$ is a Delsarte optimal anticode,
 \item $a_{mR}(\mC)=R$,
 \item for all $r \in [mR]$ we have $a_r(\mC)=\lceil r/m \rceil$.
\end{enumerate}
In particular, the Delsarte generalized
weights of a Delsarte optimal anticode $\mC \subseteq \mbox{Mat}(k\times m,
\F_q)$ only
depend on $k$, $m$ and $\mbox{maxrk}(\mC)$.
\end{corollary}

\begin{proof}
 Assume that $\mC$ is an optimal anticode. By Theorem \ref{propr3}, for all $r
\in [mR]$ we have
 $a_r(\mC) \ge \lceil r/m \rceil$. Let $r \in [mR]$. Since $\lceil r/m \rceil
\le \lceil mR/m \rceil=R$, by iterating Lemma \ref{sottoa} we can find an
optimal anticode $\mA \subseteq \mC$ with $\dim_{\F_q}(\mA)=m \lceil r/m
\rceil$. We have $\dim_{\F_q}(\mA \cap \mC)=\dim_{\F_q}(\mA)=m \lceil r/m
\rceil$, and so $a_r(\mC) \le \lceil r/m \rceil$. This proves
$(1) \Rightarrow (3)$. It is clear that $(3)$ implies $(2)$.
Let us prove $(2) \Rightarrow (1)$. Assume
$a_{mR}(\mC)=R$. By definition, there exists an optimal anticode
$\mA \in \mA_q^D(k,m)$ such that $\dim_{\F_q}(\mA)=mR$ and $\dim_{\F_q}(\mA
\cap \mC) \ge mR$. Since $\dim_{\F_q}(\mC)=mR$, we have $\mA=\mC$.
In particular, $\mC \in \mA_q^D(k,m)$.
\end{proof}

\section{Delsarte generalized weights and duality} \label{sss6}

In this section we give the definition of Delsarte dual code, and show 
that the Delsarte generalized weights of a code and the 
 Delsarte generalized weights of the dual code determine each others.
 We first recall the 
 analogous definitions and results for linear and Gabidulin codes.

The \textbf{dual} of a linear code $C\subseteq \F_q^n$
is   $C^\perp:= \{ v \in \F_q^n : \langle c, v \rangle =0 \mbox{ for all } c \in
C\}\subseteq \F_q^n$, where
  $\langle \cdot , \cdot \rangle$ is the standard inner product of $\F_q^n$.
The \textbf{dual} of a Gabidulin code $C$ is the Gabidulin code 
   $C^\perp:= \{ v \in \F_{q^m}^k : \langle c, v \rangle =0 \mbox{ for all } c
\in C\}\subseteq \F_{q^m}^k$, where
  $\langle \cdot , \cdot \rangle$ is the standard inner product of $\F_{q^m}^k$.

\begin{theorem}[\cite{wei}, Theorem 3, and \cite{ducoat}] \label{vecchidual}
 The following hold.
 \begin{enumerate}
  \item Let $C \subseteq \F_q^n$ be a linear code of dimension $1 \le t <
n$ over $\F_q$. The
  generalized Hamming weights of $C$ and $C^\perp$ determine each
other.
  \item Let $C \subseteq \F_{q^m}^k$ be a Gabidulin code of dimension $1 \le t <
k$ over $\F_{q^m}^k$. The
  generalized rank weights of $C$ and $C^\perp$ determine each other.
  \end{enumerate}
\end{theorem}

The goal of this section is to establish the analogue of Theorem
\ref{vecchidual} for Delsarte codes and
Delsarte generalized weights. We will use the notion of duality
in $\mbox{Mat}(k \times m, \F_q)$ proposed in the context of
coding theory by Delsarte in \cite{del1}.
Recall that the \textbf{trace-product} of matrices $M,N \in \mbox{Mat}(k\times
m, \F_q)$ is  
 $\langle M , N \rangle := \mbox{Tr}(MN^t)$.
 One can easily check that the map
$\langle \cdot , \cdot \rangle: \mbox{Mat}(k \times m, \F_q) \times \mbox{Mat}(k
\times m, \F_q) \to \F_q$
is symmetric, bilinear and non-degenerate.

\begin{definition} \label{dualdel}
 Let $\mC \subseteq \mbox{Mat}(k \times m,
\F_q)$ be a Delsarte code. The \textbf{dual} of $\mC$
 is the Delsarte code  
$\mC^\perp:= \{ N \in \mbox{Mat}(k \times m, \F_q) : \langle M, N \rangle =0 
  \mbox{ for all } M \in \mC\}\subseteq \mbox{Mat}(k \times m, \F_q)$.
\end{definition}

The following lemma summarizes some well-known properties of the
dual code.
The proof is left to the reader.

\begin{lemma}\label{proprduale}
 Let $\mC,\mD \subseteq \mbox{Mat}(k \times m, \F_q)$ be $\F_q$-subspaces. We
have
$$(\mC^\perp)^\perp=\mC, \ \ \ \ \ \dim_{\F_q}(\mC^\perp)= km-\dim_{\F_q}(\mC),
\ \ \ \ \ 
(\mC \cap \mD)^\perp=\mC^\perp + \mD^\perp.$$
\end{lemma}

A crucial property of the set of Delsarte optimal anticodes is that it is
preserved by duality.

\begin{theorem}[\cite{io}, Theorem 54] \label{dan}
 Let $\mA \in \mbox{Mat}(k \times m, \F_q)$ be a Delsarte code. We have
 $\mA \in \mA_q^D(k \times m)$ if and only if $\mA^\perp \in \mA_q^D(k
\times m)$.
\end{theorem}

The theorem that we now present describes how the Delsarte generalized weights
of a code $\mC$ relate to the
Delsarte generalized weights of the dual code $\mC^\perp$. We will obtain as a
corollary the main result
of this section.

\begin{theorem} \label{diverso}
 Let $\mC \subseteq \mbox{Mat}(k \times m, \F_q)$
be a Delsarte code of dimension $1 \le t \le km-1$. Assume that
$p,i,j \in \Z$ satisfy:
$$1 \le p+im \le km-t \ \ \ \ \mbox{ and } \ \ \ \ 1 \le p+t+jm \le t.$$
Then $a_{p+im}(\mC^\perp) \neq k+1-a_{p+t+jm}(\mC)$.
\end{theorem}

\begin{proof}
 Define $r:=p+im$ and $s:=t+r-m\cdot a_r(\mC^\perp)$. By Theorem \ref{propr3} we
have 
$a_r(\mC^\perp) \ge r/m$, and so $s \le t$.
 We split
the proof into
two parts. All dimensions  are over $\F_q$.
\begin{enumerate}
 \item Assume $p+t+jm \le s$. Since $p+t+jm \ge
1$, we have
 $1 \le p+t+jm \le s \le t$. Let $\mA \in
\mA^D_q(k,m)$ with 
$\dim(\mA \cap \mC^\perp) \ge r$ and $\dim(\mA)=m \cdot a_r(\mC^\perp)$.
By Lemma \ref{proprduale} we have
\begin{eqnarray*}
 r \le \dim(\mA \cap \mC^\perp) &=& \dim(\mA)+\dim(\mC^\perp)-\dim(\mA +
\mC^\perp) \\
&=& m \cdot a_r(\mC^\perp)+(km-t)-(km- \dim(\mA^\perp \cap \mC)) \\
&=& m \cdot a_r(\mC^\perp) -t + \dim(\mA^\perp \cap \mC).
\end{eqnarray*}
This implies $s=t+r-m\cdot a_r(\mC^\perp) \le \dim(\mA^\perp \cap \mC)$.
Therefore by
Theorem \ref{dan} we have
$a_s(\mC) \le \dim(\mA^\perp)/m=(km-\dim(\mA))/m=(km-m \cdot
a_r(\mC^\perp))/m=k-a_r(\mC^\perp)$,
i.e., $a_r(\mC^\perp) \le k-a_s(\mC)$. Since $p+t+jm \le s$, by Theorem
\ref{propr3} we have 
$a_s(\mC) \ge
a_{p+t+jm}(\mC)$. As a consequence,
$a_r(\mC^\perp) \le k-a_s(\mC) \le k-a_{p+t+jm}(\mC)< k+1-a_{p+t+jm}(\mC)$, and
the result follows.

\item Now assume $p+t+jm > s$, i.e., $i-j < a_r(\mC^\perp)$. Let 
$\varepsilon >0$ with $i-j=a_r(\mC^\perp)-\varepsilon$.  By definition
of $r$ we have
\begin{eqnarray*}
 p+t+jm &=& r-im+t+jm \\
&=& r-(i-j)m+t \\
&=& r-(a_r(\mC^\perp)-\varepsilon)m+t \\
&=& t+r-m \cdot a_r(\mC^\perp)+\varepsilon m \\
&=& s+\varepsilon m.
\end{eqnarray*}
Assume by contradiction
$a_r(\mC^\perp)=k+1-a_{p+t+jm}(\mC)$, i.e., $a_r(\mC^\perp)=k+1-a_{s+\varepsilon
m}(\mC)$.
Let $\mA \in \mA^D_q(k,m)$ with $\dim(\mA \cap \mC) \ge s+\varepsilon m$ and
$\dim(\mA)=m \cdot a_{s+\varepsilon m}(\mC)=m(k+1-a_r(\mC^\perp))$.
By Lemma \ref{proprduale} we have
\begin{eqnarray*}
 s+\varepsilon m &\le& \dim(\mA \cap \mC) \\
&=& \dim(\mA)+\dim(\mC)-\dim(\mA+\mC) \\
&=& m(k+1-a_r(\mC^\perp)) + t -(km-\dim(\mA^\perp \cap \mC^\perp)) \\
&=& m-m \cdot a_r(\mC^\perp)+t+\dim(\mA^\perp \cap \mC^\perp).
\end{eqnarray*}
Since $s=t+r-m\cdot a_r(\mC^\perp)$, the inequality above can be re-written
as  $\dim(\mA^\perp \cap \mC^\perp) \ge r+\varepsilon m -m$.
By Theorem \ref{dan}, $\mA^\perp \in \mA^D_q(k,m)$, and so
$m \cdot a_{r+\varepsilon m -m}(\mC^\perp) \le \dim(\mA^\perp)$. On the other
hand, by
Lemma \ref{proprduale} 
we have
$$\dim(\mA^\perp)=km-\dim(\mA)=km-m(k+1-a_r(\mC^\perp))=m(a_r(\mC^\perp)-1).$$
It follows $m \cdot a_{r+\varepsilon m -m}(\mC^\perp) \le
\dim(\mA^\perp)=m(a_r(\mC^\perp)-1)$,
i.e., $a_{r+\varepsilon m -m}(\mC^\perp) \le a_r(\mC^\perp)-1$. Since
$\varepsilon >0$, we have
$r+\varepsilon m -m \ge r$. Hence by
Theorem \ref{propr3} we have $a_r(\mC^\perp) \le a_{r+\varepsilon m
-m}(\mC^\perp) \le a_r(\mC^\perp)-1$,
a contradiction. \qedhere
\end{enumerate}
\end{proof}

We now present the main result of this section, which is the analogue of Theorem
\ref{vecchidual}
for Delsarte codes.
 Let $1 \le k \le m$ be integers, and let $\mC \subseteq \mbox{Mat}(k \times
m, \F_q)$ be a Delsarte code of dimension $1 \le t \le km$. For any
$s \in \Z$, we define the \textbf{$s$-weight sets} of $\mC$  by
\begin{eqnarray*}
 W_s(\mC) &:=& \{ a_{s+im}(\mC) : i \in \Z, \ 1 \le s+im \le t \}, \\
 \overline{W}_s(\mC) &:=& \{ k+1-a_{s+im}(\mC) : i \in \Z, \ 1 \le s+im \le t
\}.
\end{eqnarray*}
The following result holds.

\begin{corollary} \label{determinano}
 Let $\mC \subseteq \mbox{Mat}(k \times m, \F_q)$
be a Delsarte code of dimension $1 \le t \le km-1$. For any integer $1 \le p
\le m$ we have $W_p(\mC^\perp)=[k]\setminus \overline{W}_{p+t}(\mC)$. In
particular,
the Delsarte generalized
weights of $\mC$ completely determine the Delsarte generalized weights of
$\mC^\perp$.
\end{corollary}

\begin{proof}
By Theorem \ref{diverso} we have $W_p(\mC^\perp) \cap
\overline{W}_{p+t}(\mC) = \emptyset$, and parts (\ref{propr3_1}),
(\ref{propr3_2}) and (\ref{propr3_3}) of Theorem \ref{propr3} imply
$W_p(\mC^\perp) \cup
\overline{W}_{p+t}(\mC) \subseteq [k]$. Hence it suffices to show that
$|W_p(\mC^\perp)| +
|\overline{W}_{p+t}(\mC)|=k$.

By part (\ref{propr3_4}) of Theorem \ref{propr3} the generalized weights 
$a_{p+im}(\mC^\perp)$, for $i \in \Z$ with $1 \le p+im \le km-t$, are
distinct. Therefore we have
\begin{equation} \label{primaf}
 |W_p(\mC^\perp)|= |\{ i \in \Z : \lceil (1-p)/m \rceil \le i \le
\lfloor (km-t-p)/m \rfloor \}|.
\end{equation} 
For the same reason, the
generalized weights
$a_{p+t+im}(\mC)$, for $i \in \Z$ with $1 \le p+t+im \le t$, are
also distinct, and so
\begin{equation} \label{secondaf}
 |\overline{W}_{p+t}(\mC)|= |\{ i \in \Z : \lceil (1-p-t)/m
\rceil \le i \le
\lfloor -p/m \rfloor \}|.
\end{equation}
Since $1 \le p \le m$, we have $\lceil (1-p)/m \rceil=0$ and
$\lfloor -p/m \rfloor=-1$.
Thus equations (\ref{primaf}) and (\ref{secondaf}) can be written as 
$$|W_p(\mC^\perp)|= \lfloor (km-t-p)/m \rfloor+1, \ \ \
\ \ \ |\overline{W}_{p+t}(\mC)|= - \lceil (1-p-t)/m
\rceil.$$
Therefore it suffices to show
\begin{equation} \label{suff}
 \lfloor (km-t-p)/m \rfloor  - \lceil (1-p-t)/m \rceil =k-1.
\end{equation}
Write $t+p=Am+B$ with $0 \le B \le m-1$. If $B=0$ then
$\lfloor (km-t-p)/m \rfloor=k-A$ and $\lceil (1-p-t)/m \rceil = -A+1$. If
$0 <B \le m-1$ then $\lfloor (km-t-p)/m \rfloor=k-A-1$ and $\lceil (1-p-t)/m
\rceil = -A$. This shows identity (\ref{suff}).

To prove the second part of
the statement, observe that by part (\ref{propr3_4}) of Theorem \ref{propr3} the
generalized weights of $\mC^\perp$ in $W_p(\mC^\perp)$ are ordered integers. 
Hence by the first part of the statement they are determined by the
set $\overline{W}_{t+p}(\mC)$. 
The result now follows from the fact that
any $a_r(\mC^\perp)$, $1 \le r \le km-t$, belongs to exactly one set
$W_p(\mC^\perp)$,
for some $1 \le p \le m$.
\end{proof}

\begin{example}
 Let e.g. $q=5$ and $k=m=3$. Let $\mC \subseteq \mbox{Mat}(3 \times 3, \F_5)$ be
the code
 generated over $\F_5$ by the two matrices
 $$\begin{bmatrix}
    1 & 0 & 0 \\ 0 & 0 & 0 \\ 0 & 0 & 0
   \end{bmatrix}, \ \ \ \
   \begin{bmatrix}
    0 & 0 & 0 \\ 0 & 3 & 0 \\ 0 & 0 & 0
   \end{bmatrix}.
$$
We have $a_1(\mC)=1$, $a_2(\mC)=2$, and $\dim_{\F_q}(\mC^\perp)=9-2=7$. We will
compute the integers 
$$a_1(\mC^\perp), \ \ a_2(\mC^\perp), \ \ a_3(\mC^\perp), \ \ a_4(\mC^\perp), \
\ a_5(\mC^\perp), \ \ a_6(\mC^\perp), \ \ a_7(\mC^\perp)$$
employing Corollary \ref{determinano}.
Start with $p=1$. We have
$W_1(\mC^\perp)= \{ a_1(\mC^\perp), a_4(\mC^\perp), a_7(\mC^\perp)\}$
and $\overline{W}_3(\mC)=\emptyset$.
Since $a_1(\mC^\perp) < a_4(\mC^\perp) < a_7(\mC^\perp)$ and
$W_1(\mC^\perp)=[3]\setminus \overline{W}_3(\mC)$, it follows
$a_1(\mC^\perp)=1$, $a_4(\mC^\perp)=2$, $a_7(\mC^\perp)=3$.
Similarly, $W_2(\mC^\perp)= \{ a_2(\mC^\perp), a_5(\mC^\perp)\}$
and $\overline{W}_4(\mC)= \{ 3+1-a_1(\mC)\}= \{ 3 \}$.
It follows $a_2(\mC^\perp)=1$ and $a_5(\mC^\perp)=2$.
Finally, $W_3(\mC^\perp)= \{ a_3(\mC^\perp), a_6(\mC^\perp)\}$ and 
$\overline{W}_5(\mC)= \{ 3+1-a_2(\mC)\} = \{ 2\}$. Hence $a_3(\mC^\perp)=1$ and
$a_6(\mC^\perp)=3$.
Summarizing, the Delsarte
 generalized weights of $\mC^\perp$ are the integers
 $$a_1(\mC^\perp)=1, \ \ a_2(\mC^\perp)=1, \ \ a_3(\mC^\perp)=1, \ \
a_4(\mC^\perp)=2, \ \ a_5(\mC^\perp)=2, \ \ a_6(\mC^\perp)=3, 
 \ \ a_7(\mC^\perp)=3.$$
\end{example}

  Combining Theorem \ref{finer}, Lemma \ref{tett}(\ref{te2}) and \cite{io},
Theorem 21, we see that Corollary \ref{determinano} 
  generalizes the second part of Theorem \ref{vecchidual}.

\begin{remark}
    In \cite{ogg} Oggier and Sboui propose a definition of generalized
rank weights for Gabidulin codes which we now briefly describe.
Let $C \subseteq \F_{q^m}^k$
 be a non-zero Gabidulin code of dimension $1 \le t \le k$. Given an
 integer $1\le r \le t$, the \textbf{$r$-th Oggier-Sboui generalized weight} of
$C$ is 
$m'_r(C):= \min \{ \mbox{maxrk}(D) : D \subseteq C, \ \dim_{\F_{q^m}}(D)=r
\}$.  Ducoat shows in \cite{ducoat} how the Oggier-Sboui generalized weights
relate to  the generalized rank weights proposed by Kurihara, Matsumoto and
Uyematsu in
\cite{kmat}.

One may also define generalized weights for Delsarte codes in analogy with
 the generalized weights for Gabidulin codes proposed by
Oggier
and Sboui as follows.
 Given a Delsarte code $\mC \subseteq \mbox{Mat}(k
 \times m,
 \F_q)$ of dimension $1 \le t \le km$ and an integer $1 \le r \le t$,  define
  $a'_r(\mC):= \{ \mbox{maxrk}(\mD) : \mD \subseteq \mC, \
 \dim_{\F_q}(\mD)=r\}$.
 It can be proved that $a'_r(\mC) \le a_r(\mC)$ for all $r$, and that equality
 does not hold in general.  Let e.g. $q=2$, $k=2$ and $m=3$. Denote by $\mC
\subseteq \mbox{Mat}(2 \times 3,
 \F_2)$ the Delsarte code generated by the three $\F_q$-independent matrices
 $$A:=\begin{bmatrix} 1 & 0 & 0 \\ 0 & 0 & 0 \end{bmatrix}, \ \ \ 
 B:=\begin{bmatrix} 0 & 1 & 0 \\ 0 & 0 & 1 \end{bmatrix}, \ \ \ 
 C:=\begin{bmatrix} 0 & 0 & 0 \\ 1 & 0 & 0 \end{bmatrix}.$$
 The $2$-dimensional subcode $\mD \subseteq \mC$ generated by $A$ and $C$ has 
 $\mbox{maxrk}(\mD)=1$. Hence $a'_2(\mC)=1$.
 On the other hand, it can be checked that there is no Delsarte optimal
anticode
 $\mA
 \in \mA^D_2(2,3)$ with $\dim_{\F_q}(\mA)=3$ and $\dim_{\F_q}(\mA \cap
 \mC) \ge 2$. It follows $a_2(\mC)=6/3=2 \neq a'_2(\mC)$.
 
 Unfortunately, it is not true in general that 
the $a_r'$ generalized weights of a Delsarte code determine 
the $a_r'$ generalized weights of the dual code.
Let e.g. $q=2$, $k=2$ and $m=3$. Consider  the $2$-dimensional Delsarte codes
$\mC,\mD \subseteq \mbox{Mat}(k\times m, \F_2)$ defined by
$$\mC:= \langle \begin{bmatrix} 1 & 0 & 0 \\ 0 & 0 & 0  \end{bmatrix}, 
\begin{bmatrix} 0 & 0 & 0 \\ 1 & 0 & 0  \end{bmatrix}
\rangle, \ \ \ \ \  \mD:= \langle \begin{bmatrix} 1 & 0 & 0 \\ 0 & 0 & 0 
\end{bmatrix}, 
\begin{bmatrix} 0 & 1 & 0 \\ 0 & 0 & 0  \end{bmatrix}
\rangle.$$
One can  check that $a'_1(\mC)=a'_1(\mD)=1$ and  $a'_2(\mC)=a'_2(\mD)=1$.
On the other hand, we have
\begin{gather*}
 a'_1(\mC^\perp) = 1,\ \ a'_2(\mC^\perp) = 1, \ \ a'_3(\mC^\perp) = 2, \ \
a'_4(\mC^\perp) = 2, \\
 a'_1(\mD^\perp) = 1,\ \
a'_2(\mD^\perp) = 1, \ \ a'_3(\mD^\perp) = 1, \ \
a'_4(\mD^\perp) = 2.
\end{gather*}
Thus $\mC$ and $\mD$ have the same $a_r'$ generalized weights, while
$\mC^\perp$ and $\mD^\perp$ have not. 
Therefore we do not have an analogue
of Corollary \ref{determinano} for the $a_r'$ generalized weights.
\end{remark}

\section{Generalized rank weights and security drops} \label{sss7}
 
In \cite{metrics3} Silva and Kschischang propose a rank-metric coding scheme
to secure a network communication against an eavesdropper. In this paper
we are more interested in the algebraic aspects of the problem, and we do not
describe the scheme. 
In \cite{metrics3}
the authors prove that when a Gabidulin code $C \subseteq \F_{q^m}^k$ is employed in 
their scheme, the information that an eavesdropper can obtain listening at $0 \le \mu \le k$ links
of the channel is bounded by the quantity
$$\Delta_\mu(C) := \max \{ \dim_{\F_q}(V \cap C)  :  V \in \Lambda_q(k,m), \
\dim_{\F_q}(V)=\mu\}.$$
Clearly, $\Delta_\mu(C) \ge \Delta_{\mu -1}(C)$ for any Gabidulin
code
$C$ and any integer $1 \le \mu \le k$. In analogy with the theory of
generalized Hamming weights of \cite{wei}, 
we propose the following definition.

\begin{definition}
 Let $C \subseteq \F_{q^m}^k$ be a Gabidulin code. 
 An integer $1 \le \mu \le k$ is  a \textbf{worst-case security drop} for  $C$ if
 $\Delta_\mu(C) > \Delta_{\mu -1}(C)$.
\end{definition}

The following result is the analogue for Gabidulin code of \cite{wei}, Corollary
A. It shows that the generalized rank weights introduced by
Kurihara,  Matsumoto and Uyematsu
 in \cite{kmat} measure the worst-case security drops of a Gabidulin code
  employed in the scheme of \cite{metrics3}. 

\begin{theorem} \label{drrr}
 Let $C \subseteq \F_{q^m}^k$ be a Gabidulin code
 of dimension $1 \le t \le k$ over $\F_q$. 
 Fix an integer $1 \le \mu \le k$. The following are equivalent.
 \begin{enumerate}
  \item $\Delta_{\mu}(C)>\Delta_{\mu-1}(C)$, i.e., $\mu$ is a worst-case
security drop for $C$,
  \item there exists $1 \le r \le t$ with $m_r(C)=\mu$.
 \end{enumerate}
\end{theorem}

\begin{proof}
 Let us prove $(1) \Rightarrow (2)$. Take $V \in \Lambda_q(k,m)$ with
$\dim_{\F_{q^m}}(V)=\mu$ and
 $\dim_{\F_{q^m}}(V \cap C)=\Delta_\mu(C)$. We have $m_{\Delta_{\mu}(C)}(C) \le
\mu$.
 Assume by contradiction $m_{\Delta_{\mu}(C)}(C) < \mu$. By definition, there
exists $U \in \Lambda_q(k,m)$ with 
 $\dim_{\F_{q^m}}(U \cap C) \ge \Delta_\mu(C)$ and $\dim_{\F_{q^m}}(U)<\mu$.
Clearly, we can find 
 $H \supseteq U$ with $H \in \Lambda_q(k,m)$ and $\dim_{\F_{q^m}}(H)=\mu-1$.
 It follows
 $$\Delta_{\mu-1}(C) \ge \dim_{\F_{q^m}}(H \cap C) \ge \dim_{\F_{q^m}}(U \cap C)
\ge \Delta_\mu(C),$$
 a contradiction. Hence we may take $r=\Delta_{\mu}(C)$.
 Now we prove $(2) \Rightarrow (1)$. Let $1 \le r \le t$ with $m_r(C)=\mu$.
 There exists $V \in \Lambda_q(k,m)$ with $\dim_{\F_{q^m}}(V \cap C) \ge r$
 and $\dim_{\F_{q^m}}(V)=\mu$. Hence $\Delta_\mu(C) \ge r$.
 Assume by contradiction $\Delta_\mu(C)=\Delta_{\mu-1}(C)$. Let $U \in
\Lambda_q(k,m)$
with $\dim_{\F_{q^m}}(U)=\mu-1$ and $\dim_{\F_{q^m}}(U \cap C)
=\Delta_{\mu-1}(C)=\Delta_{\mu}(C)$.
By definition, $m_{\Delta_\mu(C)}(C) \le \mu-1$. Moreover, since $\Delta_\mu(C)
\ge r$,
by Theorem \ref{propr2} we have $m_{\Delta_\mu(C)}(C) \ge m_r(C)$. It follows
$\mu=m_r(C) \le m_{\Delta_\mu(C)}(C) \le \mu-1$, a contradiction. This proves
$\Delta_\mu(C) > \Delta_{\mu-1}(C)$.
\end{proof}

\section*{Acknowledgement}
I am grateful to Elisa Gorla for help in improving Theorem \ref{finer}.

\end{document}